\pdfoutput=1
\documentclass[nospthms]{svjour3}

\newif\ifarxiv
\arxivtrue   
\ifarxiv
  \usepackage[papersize={155mm,235mm},total={122mm,198mm},centering]{geometry}
\fi

\newif\ifversionnote
\versionnotetrue   

\usepackage[T1]{fontenc}
\usepackage[utf8]{inputenc}

\usepackage{amsmath,amssymb,amsfonts,amsthm}
\usepackage{mathtools}
\usepackage{braket}
\usepackage{siunitx}
\usepackage{graphicx}
\usepackage{booktabs}
\usepackage{enumitem}
\usepackage{microtype}
\usepackage[hidelinks]{hyperref}

\numberwithin{equation}{section}

\makeatletter
\def\@thmcountersep{.}

\@for\sv@released:={thm,lemma,cor,prop,defin,exa,rem}\do{%
  \expandafter\let\csname \sv@released\endcsname\relax
  \expandafter\let\csname end\sv@released\endcsname\relax
}
\makeatother

\newtheorem{thm}{Theorem}[section]
\newtheorem{lemma}[thm]{Lemma}
\newtheorem{cor}[thm]{Corollary}
\newtheorem{prop}[thm]{Proposition}
\theoremstyle{definition}
\newtheorem{defin}[thm]{Definition}
\newtheorem{exa}[thm]{Example}
\theoremstyle{remark}
\newtheorem{rem}[thm]{Remark}

\DeclareMathOperator{\Span}{span}
\DeclareMathOperator{\supp}{supp}
\DeclareMathOperator{\Ad}{Ad}
\DeclareMathOperator{\Lie}{Lie}

\title{Point-group filtering in unitary coupled-cluster ans\"atze:\\
exact criticality under the Abelian filter, freeness measured\\
near equilibrium under the full group%
\ifversionnote
\thanks{%
\textbf{Note on versions.} This is version 3 of \texttt{arXiv:2603.21009}
\cite{dasilva2026}; versions 1 and 2 remain public. Their inference of
variational incompleteness from a restricted dynamical Lie algebra does not
follow, and the 21.8~mHa they report for NH$_3$ in STO-3G was an
artifact of an incomplete excitation pool in our own code as it stood for
those versions, not a cost of the filter. Section~\ref{SS:artifact} carries both diagnoses.}%
\fi
}

\titlerunning{Point-group filtering in unitary coupled-cluster ans\"atze}

\author{Leon D.\ da Silva \and Marcelo P.\ Santos \and Md.\ Zubair}

\institute{Leon D.\ da Silva \at Departamento de Matem\'atica, Universidade
Federal Rural de Pernambuco, Recife, 52171-900, PE, Brazil \\
\email{leon.silva@ufrpe.br}
\and
Marcelo P.\ Santos \at Departamento de Matem\'atica, Universidade
Federal Rural de Pernambuco, Recife, 52171-900, PE, Brazil
\and
Md.\ Zubair \at Department of Computer Science and Engineering, Green
University of Bangladesh, Dhaka, Bangladesh}

\date{}

\begin{document}

\maketitle

\begin{abstract}
Filtering a unitary coupled-cluster ansatz by molecular point-group symmetry
has no settled variational cost for groups with degenerate irreducible
representations. We prove that the symmetry-adapted subvariety of the
amplitude space is a critical subvariety of the energy, for irreducible
representations of any dimension. The gradient along every removed direction
vanishes there, so a quasi-Newton optimisation with exact gradients started at
the reference determinant cannot distinguish the filtered from the unfiltered
ansatz. Every operator invariant
under the full point group already lies in the span of the pool retained by an
Abelian subgroup, so the filtered pool reaches at every geometry an energy no
higher than the fully invariant one reaches. A Hamiltonian-informed pool lies
inside the set the equivariant pool touches, and in a symmetry-adapted basis
inside the symmetry-filtered pool, which makes a published pair of parameter
counts a certificate for the orbital basis. Across six molecules and two basis sets,
the full-group filter changes the classical coupled-cluster correlation energy
by at most $3.4 \times 10^{-12}$ millihartree in the larger basis where the
point group is finite. The parameter count of the unitary ansatz falls from 75
to 30 for ammonia and from 65 to 21 for methane, measured against the Abelian
filter current practice uses. The freeness measured for the full-group filter
holds near equilibrium and degrades into the static-correlation regime.
\end{abstract}

\section{Introduction}
\label{S:intro}
A unitary coupled-cluster ansatz is a product of exponentials of excitation
operators, and filtering it by molecular point-group symmetry deletes from
that product every operator whose amplitude the symmetry forces to vanish. For
Abelian groups the filter is standard practice and its cost is understood. The
classical coupled-cluster literature settled the degenerate case decades ago.
\v{C}\'arsky et al.~\cite{carsky1987} proved the selection rule without
assuming one-dimensional representations, H\"aser et al.~\cite{haser1991}
wrote the group-average projector that implements it, and Stanton et
al.~\cite{stanton1991} built the implementation that became standard, while
stating that its algorithms are appropriate only for Abelian groups.
Remark~\ref{R:priority} delimits what each of those papers established, and none
of it is claimed here. What has no counterpart on the unitary side is a
statement of what the filter costs variationally when the group has
irreducible representations of dimension greater than one.

The variational quantum eigensolver \cite{peruzzo2014}, run with the unitary
coupled-cluster singles-and-doubles ansatz \cite{romero2018}, inherited the
Abelian scheme and with it the Abelian ceiling. The non-Abelian case is called
open in six places in the recent literature. Cao et al.~\cite{cao2022} exclude
it by explicit choice when projecting the amplitude space onto the totally
symmetric subspace of an Abelian subgroup. Goings et al.~\cite{goings2023}
expect savings from groups such as the full $D_{6h}$ without exploring them,
and Sakuma et al.~\cite{sakuma2026} leave the non-Abelian hardware application
as future work. Picozzi and Tennyson name ammonia, whose group $C_{3v}$ is not
Boolean, as the case where the $C_s$ subgroup will be used instead
\cite{picozzi2023}, and three years later still descend to the largest Boolean
subgroup, across nine reference molecules that all lie in $C_{2v}$, $D_{2h}$
or $C_{2h}$ \cite{picozzi2026}. He et al.~\cite{he2026} rest their non-Abelian
conclusions on numerical results alone, a hedge Section~\ref{S:certificates}
takes up in the words they state it. Greiner et al.~\cite{greiner2024}
meanwhile exercise the full $D_{6h}$ on benzene, so the barrier was never
mathematical.

Two of those papers report the same quantity for the same molecules under the
same Abelian subgroup, and the values disagree. Cao et al.\ put the difference
between the filtered and unfiltered UCCSD energies for ammonia and methane
below 0.005~mHartree, inside their own convergence threshold. He et al.\
report, for those systems and that subgroup, errors of 27.8 and
40.8~mHartree. The ratio is between five and eight thousand, and neither paper
records the conflict. A second tension is structural. Shkolnikov et
al.~\cite{shkolnikov2023} prove that a Pauli string changing the symmetry
sector has identically vanishing gradient at every step of an adaptive
algorithm and read that vanishing as a roadblock, while Cao et al.\ read the
same vanishing as an economy, and both readings follow from sector
orthogonality. No statement of variational freeness valid for irreducible
representations of dimension greater than one has been given, and no
measurement of ansatz-parameter compression under the full point group has
been reported.

We prove three statements about that amplitude space.
Theorem~\ref{T:critical} identifies the $G$-symmetric subvariety as a critical
subvariety of the energy, through the sector orthogonality that Shkolnikov et
al.~\cite{shkolnikov2023} identified for pools of Pauli strings. The gradient
along every removed direction vanishes at each point of that subvariety, and
not only at the reference determinant, for irreducible representations of any
dimension. Corollary~\ref{C:protocol} draws the operational consequence. A
quasi-Newton optimisation with exact gradients initialised at the reference
never leaves the subvariety, so the filtered and unfiltered ans\"atze follow
one trajectory and return one energy.

Theorem~\ref{T:span} places every operator invariant under the full group in
the span of the Abelian-filtered pool, so that
$\Span\mathcal{P}_G\subseteq\Span\mathcal{P}_{\mathrm{Sym}}$. That containment
runs opposite to the motivation He et al.~\cite{he2026} give for their
construction. The retained pool is the smaller one, and
Corollary~\ref{C:reach} makes it sufficient against the fully invariant pool
at every geometry, since the states its generators reach include every state
the invariant generators reach. Proposition~\ref{P:hisupport} places the
Hamiltonian-informed pool inside the support of the equivariant pool, the set
of monomials a $G$-equivariant operator can touch, and Theorem~\ref{T:span}
carries that support into the symmetry-filtered pool. In a symmetry-adapted
basis the Hamiltonian-informed pool is therefore contained in the
symmetry-filtered one (Corollary~\ref{T:hiuccsd}) and can never carry more
parameters (Corollary~\ref{C:certificate}). Remark~\ref{R:a1content} places the
inference this version withdraws in the gap between that support and the larger
set an $A_1$ shell label admits.

Each result reaches a working calculation. Theorem~\ref{T:span} and the
measurements of Section~\ref{S:numerics} together let a chemist filter
coupled-cluster amplitudes by the full point group and still recover the
unconstrained correlation energy to within \num{3.4e-12}~mHa at 6-31G,
for the molecules here whose point group is finite.
Corollary~\ref{C:certificate} turns a published pair of parameter counts into a
test of the orbital basis behind them. Section~\ref{S:numerics} locates the
geometry at which the freeness of the full-group filter ends. The mathematics
used here is standard. The
group-average projector and the Wigner--Eckart selection rule are classical,
and what is new is the object they act on, namely the amplitude space of a
variational product ansatz together with its optimisation trajectory.

The argument runs in one direction. Section~\ref{S:setup} builds the
representation theory both filters need and Section~\ref{S:abelian} settles the
Abelian case, marking the point at which its argument stops; that point is
where Theorem~\ref{T:critical} takes over (Section~\ref{S:critical}), and
Theorem~\ref{T:span} then converts the operators
Section~\ref{S:nonabelian} counts as removed into operators whose removal costs
nothing. The two certificates of Section~\ref{S:certificates} and the
measurements of Section~\ref{S:numerics} test that conversion against the
published record and locate the geometry where it ends, and
Section~\ref{S:scope} states the hypothesis that bounds all of it.
We assume familiarity with second quantisation and with character tables,
and no prior exposure to variational quantum algorithms.

\section{Symmetry-adapted orbitals, the excitation pool, and two filters}
\label{S:setup}

This section fixes the notation that gives every orbital an
irreducible-repre\-sentation label, defines the two filters that the rest of the
paper compares, and states the three hypotheses the later theorems assume. The
electronic-structure setting is standard \cite{mcardle2020}, and we develop the
representation theory here rather than cite it, because a reader who takes the
labels on trust cannot see where the two filters part company.

\begin{defin}[Adapted orbitals and the group action]
\label{D:orbitals}
Let $G$ be a finite molecular point group acting on the one-particle space.
Each spatial orbital $\phi_p$ carries a composite label $p=(\lambda,i,\mu)$,
where $\lambda$ names an irreducible representation $\rho_\lambda$ of $G$ of
dimension $d_\lambda$, the index $i$ names the shell, and
$\mu\in\{1,\dots,d_\lambda\}$ names the component. Writing $\hat a^\dagger_p$
and $\hat a_p$ for the associated creation and annihilation operators, the
action of $g\in G$ on the Fock space is the adjoint action of a unitary
$\hat R(g)$,
\begin{equation}\label{E:action}
  \hat R(g)\,\hat a^\dagger_{\lambda,i,\mu}\,\hat R(g)^{-1}
  = \sum_{\nu=1}^{d_\lambda} D^{(\lambda)}_{\nu\mu}(g)\,
    \hat a^\dagger_{\lambda,i,\nu},
\end{equation}
where $D^{(\lambda)}(g)$ is the matrix of $g$ in the basis of
$\rho_\lambda$. We use $i,j$ for occupied orbitals, $a,b$ for virtual
orbitals, and $p,q,r,s$ for general indices, and write
$N_{\mathrm{occ}}(\lambda)$ and $N_{\mathrm{vir}}(\lambda)$ for the numbers
of occupied and virtual shells transforming under $\rho_\lambda$.
\end{defin}

The composite label makes degeneracy visible, because \eqref{E:action} moves
the components of a shell among themselves while leaving $\lambda$ and $i$
fixed. A code that carries only Abelian irrep labels has no $\mu$ to record and
reads the $d_\lambda$ components as inequivalent orbitals, which is the entire
source of the phenomenon studied in Section~\ref{S:nonabelian}. Standard electronic-structure codes work in the
largest Abelian subgroup they support, which leaves a single reflection plane
for ammonia in $C_{3v}$ and three orthogonal twofold axes for methane in
$T_d$.

\begin{defin}[Adapted characters and excitation monomials]
\label{D:characters}
Let $N\subseteq G$ be a maximal Abelian subgroup whose characters are real,
and let the orbitals be adapted to $N$, so that
\begin{equation}\label{E:char}
  U(n)\,\hat a^\dagger_p\,U(n)^\dagger = \chi_p(n)\,\hat a^\dagger_p,
  \qquad \chi_p(n)\in\{\pm1\},\quad n\in N .
\end{equation}
For an excitation monomial $\tau_k$, a product of creation and annihilation
operators on adapted orbitals, set $\chi_k(n)=\prod\chi(n)$ over all orbital
indices of $\tau_k$. Because the characters are real, $\tau_k^\dagger$
carries the same character, so the anti-Hermitian generator
$\kappa_k:=\tau_k-\tau_k^\dagger$ satisfies
$U(n)\,\kappa_k\,U(n)^\dagger=\chi_k(n)\,\kappa_k$.
\end{defin}

\begin{defin}[Hamiltonian]
\label{D:hamiltonian}
The electronic Hamiltonian in second quantisation is
\begin{equation}\label{E:H}
  \mathcal{H}=\sum_{pq}h^p_q\,\hat a^\dagger_p\hat a_q
  +\tfrac12\sum_{pqrs}h^{pq}_{rs}\,
   \hat a^\dagger_p\hat a^\dagger_q\hat a_s\hat a_r ,
\end{equation}
with $h^p_q$ and $h^{pq}_{rs}$ the one- and two-electron integrals in the
adapted basis.
\end{defin}

\begin{defin}[Product ansatz and parameter partition]
\label{D:ansatz}
Let $\mathcal{P}_{\mathrm{UCCSD}}$ be the pool of all particle-number
preserving one- and two-body excitation operators. The trial state is the
product
\begin{equation}\label{E:ansatz}
  \ket{\psi(\theta)}=U(\theta)\ket{\Phi_0},
  \qquad
  U(\theta)=\prod_{k=1}^{L}\exp\bigl(\theta_{p(k)}c_k\kappa_k\bigr),
\end{equation}
where $p(k)$ assigns generators to shared parameters. We assume the tying is
\emph{parity consistent}: $\chi_k=\chi_{k'}$ whenever $p(k)=p(k')$. The
energy is $E(\theta)=\bra{\Phi_0}U(\theta)^\dagger\mathcal{H}U(\theta)
\ket{\Phi_0}$. Partition the parameters as
\begin{equation}\label{E:partition}
  S=\{m:\chi^{(m)}\equiv1 \text{ on } N\},
  \qquad A=S^{\mathsf c},
  \qquad V:=\{\theta:\theta_A=0\}.
\end{equation}
\end{defin}

Setting $\theta_A=0$ turns every factor carrying a parameter in $A$ into the
identity. The filtered ansatz is not a second construction with an energy
function of its own, but the ansatz of \eqref{E:ansatz} restricted to the
linear subvariety $V$ of \eqref{E:partition}. Every comparison made below is
therefore a comparison of the single function $E$, evaluated on $V$ and on the
whole domain.

\begin{defin}[The two filters]
\label{D:filters}
The \emph{$G$-equivariant} pool is
$\mathcal{P}_G=\{T:\Ad(g)T=T \text{ for all } g\in G\}$, where
$\Ad(g)T=\hat R(g)\,T\,\hat R(g)^{-1}$. The \emph{Abelian} pool, which is
what symmetry-adapted unitary coupled cluster implements \cite{cao2022}, is
$\mathcal{P}_{\mathrm{Sym}}=\{\tau_k:\chi_k\equiv1 \text{ on } N\}$: a single
excitation $T_{ia}$ is retained when $\chi_a\chi_i^{-1}$ is the trivial
character of $N$, and a double $T_{ijab}$ when
$\chi_a\chi_b\chi_i^{-1}\chi_j^{-1}$ is.
\end{defin}

Three hypotheses carry every result below. The first asks that the Hamiltonian
commute with the group, the second that the reference determinant be invariant
under it, and the third that parameters be tied only within a fixed character
class. Hypothesis \ref{H:A2} holds for a closed-shell reference in adapted
orbitals, because each doubly occupied orbital contributes $\chi_p^2=1$ to the
character of the determinant. It is also the hypothesis that bounds the scope
of everything that follows, and Section~\ref{S:scope} records where it fails.

\begin{enumerate}[label=(A\arabic*),leftmargin=2.6em]
  \item\label{H:A1} $[\mathcal{H},U(g)]=0$ for every $g\in G$.
  \item\label{H:A2} $U(g)\ket{\Phi_0}=\ket{\Phi_0}$ for every $g\in G$.
  \item\label{H:A3} The tying in \eqref{E:ansatz} is parity consistent.
\end{enumerate}

Hypothesis \ref{H:A3} is a property of the standard singlet tying. Generators
that share a parameter there are built from the same spatial orbitals and carry
the same character, so the hypothesis constrains nothing in practice.

\section{Abelian groups: the two filters coincide}
\label{S:abelian}

When every irreducible representation of $G$ is one-dimensional, no orbital
shell has more than one component, and the two filters of
Definition~\ref{D:filters} select the same pool. We state the Abelian case
because it marks the boundary of the distinction the paper draws, not because
it contributes a result. Nothing in the non-Abelian discussion of
Section~\ref{S:nonabelian} can be inferred from experience with the Abelian
one.

\begin{thm}[Equivalence for Abelian groups]
\label{T:abelian}
Let $G$ be Abelian. Then
$\mathcal{P}_{\mathrm{Sym}}=\mathcal{P}_G$, and every anti-Hermitian
$G$-equivariant operator in the singles-and-doubles space is a linear
combination of elements of $\mathcal{P}_{\mathrm{Sym}}$.
\end{thm}

\begin{proof}
Every irreducible representation of an Abelian group is one-dimensional, so
$d_\lambda=1$ for all $\lambda$ and each orbital shell carries a single
one-dimensional label. Taking $N=G$ in Definition~\ref{D:filters}, the
condition $\chi_a\chi_i^{-1}=\chi_0$ is the condition that
$\rho(a)\otimes\rho(i)^*$ contain the trivial representation, which is the
$G$-equivariant rule. The two pools are defined by the same condition.
\end{proof}

\begin{rem}[Priority]
\label{R:priority}
Theorem~\ref{T:abelian} is the unitary form of a statement that classical
coupled-cluster theory has carried since 1987. {\v C}{\'a}rsky et
al.~\cite{carsky1987} proved that the pseudo-integrals constituting the
double-excitation amplitudes retain every symmetry property of ordinary
integrals, and vanish unless the product of the four orbitals contains a
totally symmetric component. Their proof uses the full matrix elements of the
irreducible representations and never assumes one-dimensional irreps, and the
molecules they treat include NH$_3$ in $C_{3v}$. H{\"a}ser et
al.~\cite{haser1991} generalised the construction to an arbitrary discrete
point group by applying the totally symmetric projector
$\lvert G\rvert^{-1}\sum_R D(R)$ to quadruples of orbitals, and credited
\cite{carsky1987} with the non-Abelian extension. Stanton et
al.~\cite{stanton1991} turned the rule into the direct-product decomposition
that became standard, declared it appropriate only for Abelian groups, and
there recorded the generalisation following \cite{carsky1987} as known and
unimplemented. Cao et al.~\cite{cao2022} carried the Abelian
statement to the unitary ansatz, and observed there what does not transfer from
the classical setting, since the de-excitation operator blocks the natural
truncation of coupled cluster. Section~\ref{S:critical} therefore argues
variationally instead of by amplitude vanishing.
\end{rem}

\section{Equivariance and the critical subvariety}
\label{S:critical}

The mechanism of this section is an action of the symmetry group on the parameters
of the ansatz. Each element of $N$ flips the sign of every amplitude carrying a
non-trivial character and leaves the remaining amplitudes fixed. The energy is
invariant under that action, so its restriction to any flipped direction is an even
function of one variable, and an even function has vanishing derivative at the
origin. Lemma~\ref{L:equivariance}, Proposition~\ref{P:invariance}
and Theorem~\ref{T:critical} state the action, the invariance and the vanishing, in
that order.

\begin{lemma}[Equivariance of the ansatz]
\label{L:equivariance}
Assume \ref{H:A3}. For every $n\in N$,
\begin{equation}\label{E:equivariance}
  U(n)\,U(\theta)\,U(n)^\dagger=U(\sigma_n\theta),
  \qquad (\sigma_n\theta)_m=\chi^{(m)}(n)\,\theta_m .
\end{equation}
\end{lemma}

\begin{proof}
Conjugation by a unitary is an algebra automorphism, so it distributes over
the product in \eqref{E:ansatz} and commutes with the exponential. Factor by
factor,
\[
  U(n)\,e^{\theta c_k\kappa_k}\,U(n)^\dagger
  = e^{\theta c_k\,U(n)\kappa_kU(n)^\dagger}
  = e^{\theta c_k\chi_k(n)\kappa_k},
\]
using Definition~\ref{D:characters} in the last step. By \ref{H:A3} the
character is constant on each parameter class, so the sign
$\chi_k(n)$ depends only on $p(k)$ and can be absorbed into the parameter.
\end{proof}

\begin{prop}[Invariance of the energy]
\label{P:invariance}
Assume \ref{H:A1}--\ref{H:A3}. Then $E(\sigma_n\theta)=E(\theta)$ for every
$n\in N$ and every $\theta$.
\end{prop}

\begin{proof}
By Lemma~\ref{L:equivariance}, $U(\sigma_n\theta)=U(n)U(\theta)U(n)^\dagger$,
so that
\begin{align*}
  U(\sigma_n\theta)^\dagger\,\mathcal{H}\,U(\sigma_n\theta)
  &= U(n)\,U(\theta)^\dagger\,
    \underbrace{U(n)^\dagger\mathcal{H}U(n)}_{=\,\mathcal{H}\text{ by }
    \ref{H:A1}}\,U(\theta)\,U(n)^\dagger \\
  &= U(n)\,U(\theta)^\dagger\mathcal{H}U(\theta)\,U(n)^\dagger .
\end{align*}
Taking the expectation value in $\ket{\Phi_0}$ and using \ref{H:A2} on both
sides gives the claim.
\end{proof}

\begin{thm}[The symmetric subvariety is critical]
\label{T:critical}
Assume \ref{H:A1}--\ref{H:A3}. For every $\theta_S$ and every $m\in A$,
\begin{equation}\label{E:critical}
  \left.\frac{\partial E}{\partial\theta_m}\right|_{(\theta_S,0)}=0 .
\end{equation}
The subvariety $V$ of \eqref{E:partition} is therefore a critical
subvariety of $E$, not merely a set containing critical points.
\end{thm}

\begin{proof}
The argument is a parity argument, and the point is that it runs at an
arbitrary point of $V$ rather than only at the reference. Since $m\in A$,
there is $n\in N$ with $\chi^{(m)}(n)=-1$. Put
$F=\{m':\chi^{(m')}(n)=-1\}$, so that $m\in F\subseteq A$, and define
$\psi(t)=E(\theta_S,\,t\,e_m)$. At the point $(\theta_S,t\,e_m)$ every
coordinate in $A$ other than $m$ vanishes, and the coordinates in $S$ are
fixed by $\sigma_n$; hence
$\sigma_n(\theta_S,t\,e_m)=(\theta_S,-t\,e_m)$. Proposition~\ref{P:invariance}
now gives $\psi(t)=\psi(-t)$, so $\psi$ is even and $\psi'(0)=0$.
\end{proof}

\begin{rem}
\label{R:critical}
The vanishing in \eqref{E:critical} holds at every point of $V$, with the reference
determinant one point among them, so the theorem constrains the whole optimisation
trajectory rather than the first step. The hypotheses that carry it are the product
form of \eqref{E:ansatz}, parity-consistent tying of the amplitudes, and a reference
invariant under $G$. The theorem says nothing about whether the two infima agree, and
Proposition~\ref{P:enclosure} addresses that separate question instead. Equality can
fail, for the same reason that an unrestricted Hartree-Fock solution can lie below the
restricted one. The mechanism at work is sector orthogonality, which
\cite{shkolnikov2023} use to prove that gradients vanish along the whole trajectory of
an adaptive algorithm, and which equivariant quantum circuits build in by construction
\cite{nguyen2024}. Theorem~\ref{T:critical} extends that argument in two directions,
because \cite{shkolnikov2023} fix the reference point and assume that every generator
transforms as a one-dimensional irreducible representation, while the statement here
requires neither.
\end{rem}

\begin{cor}[Initialisation at the reference does not distinguish the two
ans\"atze]
\label{C:protocol}
Assume \ref{H:A1}--\ref{H:A3}. Consider a quasi-Newton optimisation with
exact gradients, initial approximate Hessian $B_0=I$ and $\theta^{(0)}=0$,
whose step lengths are produced by a line search and whose secant update is
skipped at any step where the curvature condition $y^{\top}s>0$ fails. Its
iterates remain in $V$ at every step, and the unrestricted ansatz and the
ansatz restricted to $S$ therefore follow the same trajectory and return the
same energy.
\end{cor}

\begin{proof}
We induct on the iteration index, carrying three properties: the iterate
satisfies $\theta^{(j)}\in V$; the approximate Hessian $B_j$ is positive
definite; and $B_j$ preserves $V$ and acts as the identity on $V^\perp$. The
base case is $\theta^{(0)}=0\in V$ with $B_0=I$. For the step,
Theorem~\ref{T:critical} gives $\nabla E(\theta^{(j)})\in V$, so the search
direction $d_j=-B_j^{-1}\nabla E(\theta^{(j)})$ lies in $V$. Since $V$ is a
linear subspace, $\theta^{(j+1)}=\theta^{(j)}+\alpha_j d_j$ lies in $V$ for
every step length $\alpha_j$, and the conclusion about the iterates is
therefore independent of which line search produced that step. Both
$s=\theta^{(j+1)}-\theta^{(j)}$ and
$y=\nabla E(\theta^{(j+1)})-\nabla E(\theta^{(j)})$ then lie in $V$. If $s=0$
the iteration has terminated. Otherwise, when $y^{\top}s>0$, positive
definiteness of $B_j$ gives $s^{\top}B_j\,s>0$, so both denominators of the
BFGS update
\[
  B_{j+1}=B_j+\frac{y\,y^{\top}}{y^{\top}s}
  -\frac{B_j\,s\,s^{\top}B_j}{s^{\top}B_j\,s}
\]
are non-zero, the update is well defined and preserves positive definiteness
by the standard argument, and it adds only operators supported on
$V\otimes V$, which preserves the hypothesis on $B$. When $y^{\top}s\le0$ the
update is skipped and $B_{j+1}=B_j$ carries the hypothesis unchanged.
\end{proof}

Corollary~\ref{C:protocol} constrains how a comparison between the two ans\"atze can
be reported. An experiment that initialises every parameter at zero and optimises
with a quasi-Newton method cannot separate them, whatever quantity it measures. A
separation reported under that protocol is evidence about the implementation and not
about the ansatz, and Section~\ref{S:certificates} identifies which part of the
implementation.

The measurements of Section~\ref{S:numerics} use the limited-memory variant of the
same method, with analytic gradients and no bounds in force. Its two-loop recursion
builds each search direction from the stored pairs $(s_i,y_i)$, which lie in $V$ by
the same induction, so that direction lies in $V$ as well and the corollary covers
those runs. The curvature test governing which pairs the variant stores is the one
assumed above.

\section{Non-Abelian groups: what the Abelian filter removes}
\label{S:nonabelian}

For a non-Abelian group the Abelian filter keeps fewer operators than the full group
would admit, and the natural expectation is that it also reaches fewer states. This
section separates the two questions, and the expectation fails. The number of
operators falls, by the amount Proposition~\ref{P:deficit} computes, while the span of
the fully invariant operators is untouched, which is Theorem~\ref{T:span}.

\begin{lemma}[Splitting of a faithful irrep under Abelian restriction]
\label{L:splitting}
Let $G$ be non-Abelian with a \emph{faithful} irreducible representation
$\rho_\lambda$ of dimension $d_\lambda>1$, and let $N\leq G$ be the maximal
Abelian subgroup used by the filter. Then
\begin{equation}\label{E:splitting}
  \rho_\lambda\big|_N=\sigma_{\lambda,1}\oplus\cdots\oplus
  \sigma_{\lambda,d_\lambda},
\end{equation}
and at least two of the summands are distinct.
\end{lemma}

\begin{proof}
The argument is by contradiction. Since $N$ is Abelian, the restriction
decomposes completely into one-dimensional representations
$\sigma_{\lambda,\mu}$ \cite[Ch.~2]{serre1977}, which is
\eqref{E:splitting}. Suppose all of them coincided, say
$\sigma_{\lambda,\mu}=\sigma$ for every $\mu$. Then
$\rho_\lambda(n)=\sigma(n)I_{d_\lambda}$ for every $n\in N$, so
$\rho_\lambda(n)$ commutes with $\rho_\lambda(g)$ for every $g\in G$.
Faithfulness carries that back to $G$ itself: $n$ commutes with every $g\in
G$, so $N\subseteq Z(G)$, and maximality of $N$ among Abelian subgroups
forces $N=Z(G)$. The centre of a non-Abelian group is never maximal Abelian,
because for any $g\notin Z(G)$ the subgroup generated by $Z(G)$ and $g$ is
Abelian and strictly larger. This contradicts the maximality of $N$.
\end{proof}

Faithfulness is a hypothesis here, not a formality. The two-dimensional irrep
$E$ of $T_d$ has kernel $D_2$, the very subgroup the filter uses, so it
restricts to the trivial character twice and its components stay
indistinguishable to the filter. Where Lemma~\ref{L:splitting} does apply,
two components of one degenerate shell receive different characters of $N$,
and orbitals that the full group cannot tell apart become distinguishable to
the filter. Both irreps that carry the molecules studied here are faithful.
The excitations that mix the components of a split shell carry a non-trivial
character and are discarded, so the pool loses the off-diagonal operators of
every such channel.
Proposition~\ref{P:deficit} counts them under a multiplicity-free hypothesis,
which is read off a character table.

\begin{prop}[What the filter removes]
\label{P:deficit}
Fix $\lambda$ with $d_\lambda>1$ and assume in addition that
$\rho_\lambda|_N$ is multiplicity free, that is, the characters
$\{\sigma_{\lambda,\mu}\}_{\mu=1}^{d_\lambda}$ of \eqref{E:splitting} are
pairwise distinct. Then for each pair of occupied and virtual shells
transforming under $\rho_\lambda$, the Abelian filter retains exactly the
$d_\lambda$ diagonal single excitations and discards
$d_\lambda(d_\lambda-1)$ of them. For doubles the criterion is equality of
character products rather than diagonality of the index tuple. The filter
retains the tuples with $\sigma_{\lambda,\nu_1}\sigma_{\lambda,\nu_2}
=\sigma_{\lambda,\mu_1}\sigma_{\lambda,\mu_2}$, and their number is, writing
$\widehat N$ for the character group of $N$,
\begin{equation}\label{E:doubles_retained}
  R_\lambda=\sum_{\chi\in\widehat N}m_\lambda(\chi)^2,
  \qquad
  m_\lambda(\chi)=\#\bigl\{(\nu,\nu'):
  \sigma_{\lambda,\nu}\sigma_{\lambda,\nu'}=\chi\bigr\},
\end{equation}
so that $d_\lambda^4-R_\lambda$ doubles are discarded. Summing over irreps and
shell pairs, the total number of single excitations removed is
\begin{equation}\label{E:total_deficit}
  \Delta_{\mathrm{singles}}
  =\sum_\lambda N_{\mathrm{occ}}(\lambda)\,N_{\mathrm{vir}}(\lambda)\,
   d_\lambda(d_\lambda-1).
\end{equation}
\end{prop}

\begin{proof}
Write $\hat E^{(ia)}_{\nu\mu}=\hat a^\dagger_{a,\nu}\hat a_{i,\mu}$ for the
$d_\lambda^2$ single excitations between the two shells. By
Definition~\ref{D:filters} the filter retains $\hat E^{(ia)}_{\nu\mu}$ when
$\sigma_{\lambda,\nu}=\sigma_{\lambda,\mu}$ in $\widehat N$. Under the
multiplicity-free hypothesis this holds exactly when $\nu=\mu$, which leaves
$d_\lambda$ operators and removes $d_\lambda^2-d_\lambda$.

Applied to the four-fold tuple the same criterion requires the product of the
two virtual characters with the inverses of the two occupied ones to be
trivial on $N$, that is,
$\sigma_{\lambda,\nu_1}\sigma_{\lambda,\nu_2}
=\sigma_{\lambda,\mu_1}\sigma_{\lambda,\mu_2}$. Grouping the
$d_\lambda^2$ ordered pairs by that product produces the multiplicities
$m_\lambda(\chi)$, and a tuple is retained when its two pairs fall in the
same group, which is the count \eqref{E:doubles_retained}.

The doubles criterion is weaker than diagonality, and the gap is not an
accident of particular groups. A tuple whose two ordered pairs agree
satisfies the product condition, and so does the tuple obtained by swapping
$\mu_1$ with $\mu_2$; together these already give $2d_\lambda^2-d_\lambda$
retained tuples, against the $d_\lambda^2$ that a diagonal criterion would
allow. Any count of the form $d_\lambda^4-d_\lambda^2$ therefore overstates
what the filter removes whenever $d_\lambda>1$.
\end{proof}

\begin{table}[t]
\centering
\caption{Single and double excitations removed by the Abelian filter per pair
of occupied and virtual shells, for the two multidimensional irreps that
carry the molecules studied here. Both restrictions are multiplicity free, as
required by Proposition~\ref{P:deficit}, which is verified by inspection of
the character tables. The doubles column follows
from~\eqref{E:doubles_retained}, not from a diagonal criterion. The counts are
what the filter removes. They are not a variational cost, by
Theorem~\ref{T:span}.}
\label{Tab:deficit}
\begin{tabular}{llcccc}
\toprule
Group & Irrep & $d_\lambda$ & Abelian subgroup & Singles removed
& Doubles removed \\
\midrule
$C_{3v}$ & $E$   & 2 & $C_s$ & 2 of 4 & 8 of 16 \\
$T_d$    & $T_2$ & 3 & $D_2$ & 6 of 9 & 60 of 81 \\
\bottomrule
\end{tabular}
\end{table}

What Table~\ref{Tab:deficit} counts is discarded operators. That count has been read
as a variational deficit, and that reading motivates the construction of
\cite{he2026}. It does not follow. The discarded operators are exactly
the ones a full $G$-invariance filter would discard as well, so the containment
between the two pools runs in the opposite direction, which Theorem~\ref{T:span}
states.

\begin{thm}[Every fully invariant operator is retained]
\label{T:span}
Let $G$ be finite, let the orbitals be adapted to $N\subseteq G$, and let
$O$ be an operator on the singles-and-doubles excitation space that is
invariant under the full group $G$. Then
\begin{equation}\label{E:span}
  O\in\Span\{\tau_k:\chi_k\equiv1 \text{ on } N\},
\end{equation}
that is, $O$ lies in the span of the Abelian pool. Equivalently,
$\Span\mathcal{P}_G\subseteq\Span\mathcal{P}_{\mathrm{Sym}}$.
\end{thm}

\begin{proof}
The proof is the uniqueness of an expansion in a basis, applied to an
operator that is invariant under a subgroup. Write $O=\sum_k c_k\tau_k$ in
the basis of adapted monomials. For $n\in N\subseteq G$, invariance of $O$
under $G$ gives in particular
\[
  O=U(n)\,O\,U(n)^\dagger=\sum_k c_k\,\chi_k(n)\,\tau_k .
\]
Since $\{\tau_k\}$ is a basis the expansion is unique, so
$c_k=c_k\chi_k(n)$ for every $k$ and every $n\in N$. Whenever
$\chi_k(n)=-1$ for some $n$, this forces $c_k=0$, and what survives is
supported on the monomials with $\chi_k\equiv1$.
\end{proof}

Containment of spans passes to the algebras the two pools generate, because the
commutator is bilinear and an iterated commutator therefore reads only the span
of its arguments. That carries the comparison from operators to states.

\begin{cor}[Containment of the reachable sets, at every geometry]
\label{C:reach}
Assume the setting of Theorem~\ref{T:span}, with $G$ acting on real adapted
orbitals. Write
$\mathfrak{g}_{\mathrm{Sym}}=\Lie\bigl(\{\kappa_k:\tau_k\in
\mathcal{P}_{\mathrm{Sym}}\}\bigr)$, with $\kappa_k=\tau_k-\tau_k^\dagger$, and
$\mathfrak{g}_G=\Lie\bigl(\{T-T^\dagger:T\in\mathcal{P}_G\}\bigr)$ for the
algebras the two pools generate under iterated commutators. Then
$\mathfrak{g}_G\subseteq\mathfrak{g}_{\mathrm{Sym}}$, so the set
$\mathcal{R}_{\mathrm{Sym}}$ of states reachable from $\ket{\Phi_0}$ by the
Abelian-filtered generators contains the set $\mathcal{R}_G$ reachable by the
$G$-equivariant ones, and
\begin{equation}\label{E:reach}
  \inf_{\mathcal{R}_{\mathrm{Sym}}}E\;\leq\;\inf_{\mathcal{R}_G}E
\end{equation}
at every geometry, with no numerical input.
\end{cor}

\begin{proof}
The algebra generated by a set of operators depends on the set only through its
real span, so it is enough to place each generator of $\mathcal{P}_G$ in the
real span of the $\kappa_k$. Let $T\in\mathcal{P}_G$ and write
$T=\sum_k c_k\tau_k$, an expansion Theorem~\ref{T:span} supports on the retained
monomials. The coefficients are real, because $G$ acts on real adapted orbitals
by orthogonal matrices and the equivariant projection
$\lvert G\rvert^{-1}\sum_R D(R)$ then has real entries. Hence
$T-T^\dagger=\sum_k c_k(\tau_k-\tau_k^\dagger)=\sum_k c_k\kappa_k$, which is the
required containment of generating sets. Containment of the algebras gives
containment of the connected groups they generate and so of the reachable sets,
and an infimum taken over the larger set cannot be the larger of the two.
\end{proof}

Both sides of~\eqref{E:reach} are reachable sets of generated groups. A product
ansatz with a fixed operator ordering and a fixed depth reaches a subset of its
own, so~\eqref{E:reach} compares the two pools rather than two circuits of a
fixed shape. Neither side is claimed to contain the exact ground state. That
comparison, against the unrestricted ansatz, is the one
Proposition~\ref{P:enclosure} and Section~\ref{S:numerics} make.

\begin{exa}[The $E$ channel of $C_{3v}$]
\label{X:c3v}
The group $C_{3v}=\{E,C_3,C_3^2,\sigma_v,\sigma_v',\sigma_v''\}$ has a
two-dimensional irreducible representation $E$. The filter uses the maximal Abelian
subgroup $C_s=\{E,\sigma_v\}$, under which $E$ splits into $A'$ on the component $x$
and $A''$ on the component $y$. The single-excitation channel $1e\to2e$ carries four
spatial operators, of which the filter keeps the two diagonal ones and discards
$\hat a^\dagger_{2e_x}\hat a_{1e_y}$ and $\hat a^\dagger_{2e_y}\hat a_{1e_x}$, whose
characters under $\sigma_v$ equal $-1$. Projecting the four onto the irreducible
components of the full group gives
\begin{align}
  T^{(A_1)} &=\tfrac12\bigl(\hat a^\dagger_{2e_x}\hat a_{1e_x}
                          +\hat a^\dagger_{2e_y}\hat a_{1e_y}\bigr),
  \label{E:TA1}\\
  T^{(A_2)} &=\tfrac12\bigl(\hat a^\dagger_{2e_x}\hat a_{1e_y}
                          -\hat a^\dagger_{2e_y}\hat a_{1e_x}\bigr),
  \label{E:TA2}\\
  T^{(E)}_1 &=\tfrac12\bigl(\hat a^\dagger_{2e_x}\hat a_{1e_x}
                          -\hat a^\dagger_{2e_y}\hat a_{1e_y}\bigr),
  \qquad
  T^{(E)}_2 =\tfrac12\bigl(\hat a^\dagger_{2e_x}\hat a_{1e_y}
                          +\hat a^\dagger_{2e_y}\hat a_{1e_x}\bigr).
  \label{E:TE}
\end{align}
The retained pair spans $T^{(A_1)}$ together with $T^{(E)}_1$, so the unique totally
symmetric operator of the channel survives the filter. The discarded pair spans
$T^{(A_2)}$ and $T^{(E)}_2$ and carries no $A_1$ component, so a full $C_{3v}$
invariance filter discards it as well. The Abelian subgroup is the less restrictive of
the two filters and never the more restrictive, which is Theorem~\ref{T:span} in the
smallest case where the two can differ.
\end{exa}

\begin{prop}[Variational enclosure at a fixed geometry]
\label{P:enclosure}
Suppose the ansatz preserves particle number and $S_z$, and suppose some
$\theta^\star\in V$ satisfies $E(\theta^\star)=E_{\mathrm{ref}}+\delta$,
where $E_{\mathrm{ref}}$ is the exact ground-state energy in the same
symmetry sector and $\delta\geq0$. Then
\begin{equation}\label{E:enclosure}
  0\;\leq\;\inf_{V}E-\inf_{\text{full}}E\;\leq\;\delta .
\end{equation}
\end{prop}

\begin{proof}
The left inequality holds because $V$ is a subset of the full parameter
domain, so the infimum over $V$ cannot be smaller. For the right inequality,
$\inf_V E\leq E(\theta^\star)=E_{\mathrm{ref}}+\delta$, while
$\inf_{\text{full}}E\geq E_{\mathrm{ref}}$ by the variational principle,
which applies because every state the ansatz generates lies in the sector of
$E_{\mathrm{ref}}$.
\end{proof}

\begin{rem}[What the smaller orbit is sufficient for]
\label{R:orbit}
The retained generators of one pair of degenerate shells do commute, and the subgroup
they generate is a torus sitting strictly inside the unitary group of that shell. The
step that does not follow is the next one, because reaching the full unitary group of
the shell is not what the variational problem asks for. The target is the ground
state, it is totally symmetric under the hypotheses of Section~\ref{S:setup}, and
Corollary~\ref{C:reach} places every state the fully invariant generators reach
inside the set the filtered ones reach. Sufficiency in that sense is sufficiency
against the invariant pool, and it holds at every geometry. Against the
unrestricted ansatz nothing of the kind is proved, and
Proposition~\ref{P:enclosure} bounds that gap instead, using the
variational principle and one exhibited point, with no representation theory in the
argument. The enclosure binds at the geometry where the point was exhibited,
Section~\ref{S:numerics} measures what happens away from it, and nothing proved here
says that the two infima are equal. Proposition~\ref{P:dla} adds that the commuting
behaviour does not extend from one pair of shells to the whole ansatz.
\end{rem}

\section{Certificates for the orbital basis, and a mechanism for the reported
values}
\label{S:certificates}

The results above say that the filter removes no fully invariant operator, and that the
energy has vanishing gradient along every direction it does remove. Parts of the
literature report a measurable penalty, so the discrepancy has to be locatable in what
was measured. This section
gives two statements a reader can check against a published table without running a
calculation, and it closes with a mechanism that accounts for the reported values.

\begin{prop}[Vanishing cross-component integrals in an adapted basis]
\label{P:integrals}
Let $N\leq G$ be Abelian and let $\{\phi_{i,\mu}\}$ be molecular orbitals
obtained by diagonalising the Fock operator within each $N$-irrep block.
Then for orbitals belonging to distinct characters,
$\chi_{i,\mu}\neq\chi_{j,\nu}$, the one-electron integrals satisfy
\begin{equation}\label{E:vanishing}
  h^{(i,\mu)}_{(j,\nu)}
  =\int\phi^*_{i,\mu}(\mathbf{x})\,\hat h\,\phi_{j,\nu}(\mathbf{x})\,
   d\mathbf{x}=0,
\end{equation}
and the two-electron integrals $h^{pq}_{rs}$ vanish whenever the character
product $\chi_p\chi_q\chi_r^{-1}\chi_s^{-1}$ is non-trivial on $N$.
\end{prop}

\begin{proof}
The Fock operator commutes with every element of $N$, so in the adapted
basis it is block diagonal with respect to the characters of $N$ and its
eigenvectors carry definite characters. The one-electron Hamiltonian
$\hat h$ likewise commutes with $N$, so its matrix elements between orbitals
of different characters vanish by orthogonality of irreducible
representations \cite[Ch.~2]{serre1977}. The same argument applied to the
Coulomb operator, which is invariant under $N$ acting simultaneously on both
coordinates, gives the two-electron statement.
\end{proof}

The pool the Hamiltonian selects goes inside the Abelian pool through a smaller
set than $\mathcal{P}_{\mathrm{Sym}}$, namely the set the equivariant pool
touches. Naming that set is what separates the operators a group average keeps
from the operators a shell label admits.

\begin{prop}[The Hamiltonian-informed pool lies in the support of the
equivariant pool]
\label{P:hisupport}
Assume \ref{H:A1} and \ref{H:A2}, and write
\begin{equation}\label{E:supp}
  \supp\mathcal{P}_G=\bigl\{\tau_k:\ \tau_k \text{ carries a nonzero
  coefficient in some } T\in\mathcal{P}_G\bigr\}
\end{equation}
for the set of excitation monomials the $G$-equivariant pool touches, the
coefficients being those of the expansion in the basis
$\mathcal{P}_{\mathrm{UCCSD}}$ of Definition~\ref{D:ansatz}. Let the
Hamiltonian-informed filter retain an
excitation $\tau_k$ when the index tuple of $\tau_k$ occurs among the terms of
$\mathcal{H}$ in \eqref{E:H}. Then
\begin{equation}\label{E:hisupport}
  \mathcal{P}_{\mathrm{Hi}}\subseteq\supp\mathcal{P}_G .
\end{equation}
\end{prop}

\begin{proof}
The Hamiltonian is not itself an element of the excitation space, so the
argument is a projection followed by one line. Let $\Pi$ keep from an operator
those monomials whose index tuple excites out of occupied and into virtual
orbitals, and discard the rest, so that $\mathcal{P}_{\mathrm{Hi}}$ is the
support of $\Pi\mathcal{H}$ by the definition just given. Hypothesis
\ref{H:A2} makes the occupied space $G$-stable, and with it the virtual space,
so $\Ad(g)$ carries an excitation monomial to a combination of excitation
monomials and $\Pi$ commutes with every $\Ad(g)$. Hypothesis \ref{H:A1} gives
$\Ad(g)\mathcal{H}=\mathcal{H}$, whence
\[
  \Ad(g)\,\Pi\mathcal{H}=\Pi\,\Ad(g)\mathcal{H}=\Pi\mathcal{H}
  \qquad\text{for every } g\in G,
\]
so that $\Pi\mathcal{H}$ is itself an element of $\mathcal{P}_G$. Every
monomial carrying a nonzero coefficient in an element of $\mathcal{P}_G$ lies
in $\supp\mathcal{P}_G$, which is \eqref{E:hisupport}.
\end{proof}

\begin{cor}[In an adapted basis, the Hamiltonian-informed pool is contained
in the Abelian pool]
\label{T:hiuccsd}
Let the orbitals be adapted to $N$ as in Definition~\ref{D:characters}. Then
every excitation the Hamiltonian-informed filter retains satisfies
$\chi_k\equiv1$ on $N$, so
\begin{equation}\label{E:hiuccsd}
  \mathcal{P}_{\mathrm{Hi}}\subseteq\supp\mathcal{P}_G
  \subseteq\mathcal{P}_{\mathrm{Sym}} .
\end{equation}
\end{cor}

\begin{proof}
Theorem~\ref{T:span} expands every element of $\mathcal{P}_G$ on the monomials
with $\chi_k\equiv1$ on $N$, so those are the only monomials the equivariant
pool touches and the second inclusion holds. The first is
Proposition~\ref{P:hisupport}.
\end{proof}

Proposition~\ref{P:integrals} reaches the same containment through the
integrals, where the argument above works through the operators. In the adapted basis
$h^p_q=0$ unless $\chi_p=\chi_q$, and $h^{pq}_{rs}=0$ unless
$\chi_p\chi_q\chi_r^{-1}\chi_s^{-1}$ is trivial, so every surviving term of
$\mathcal{H}$ has a totally symmetric index tuple and a retained excitation
carries $\chi_k\equiv1$ at once. What the route through the support adds is the
middle term of \eqref{E:hiuccsd}. That term is the set the equivariant pool
touches, and it is strictly smaller than $\mathcal{P}_{\mathrm{Sym}}$ in both
non-Abelian molecules of Section~\ref{S:numerics}.

\begin{rem}[The support is smaller than the $A_1$ content]
\label{R:a1content}
Membership of $\supp\mathcal{P}_G$ is not the condition that the shells of
$\tau_k$ admit an invariant. A pair of shells whose tensor product contains the
trivial representation can carry monomials with no component along the
invariant line, and a filter that keeps such a monomial on the strength of the
shell label keeps an operator the group average annihilates.
Example~\ref{X:c3v} exhibits the difference in its smallest instance. The four
monomials of the $1e\to2e$ channel are built from one pair of shells, whose
product contains $A_1$ once; the equivariant operators of the channel are the
multiples of $T^{(A_1)}$ in \eqref{E:TA1}, whose support is the two diagonal
monomials; the off-diagonal pair lies outside $\supp\mathcal{P}_G$ although its
shells admit the invariant. The two conditions are not the same. For an
Abelian group they do coincide, every
irreducible representation being one-dimensional, which is why the difference
is invisible in the setting where the literature works.
\end{rem}

\begin{rem}[The reverse inclusion, and why it is not stated here as a theorem]
\label{R:hisupport}
Inclusion \eqref{E:hisupport} is not an equality, and the missing direction is
of a different kind from the one proved. To assert
$\supp\mathcal{P}_G\subseteq\mathcal{P}_{\mathrm{Hi}}$ is to assert that no
integral of $\mathcal{H}$ vanishes accidentally on the equivariant support. The
support is where a generic element of $\mathcal{P}_G$ is nonzero, while
$\mathcal{P}_{\mathrm{Hi}}$ is where one particular element of it is, so the
inclusion says that the particular one is generic. Symmetry cannot supply that,
because it constrains which coefficients must vanish and never which must not.
A molecule with an accidental zero, from spatial separation or from a
basis-set artefact, would satisfy Proposition~\ref{P:hisupport} and violate its
converse. In the closed-shell cases examined in preparing this work, spanning
two non-Abelian groups, two basis sets and an Abelian control, the two sets
agreed exactly. That is evidence for genericity and not a proof of it, so the
equality is recorded here as an observation and used nowhere below.
Where the equality does hold, the Hamiltonian-informed
construction of \cite{he2026} and the full-group filter select the same
operators in an adapted basis and differ in how they parameterise them, which
makes the choice between the two a choice of parameterisation and not of reach.
\end{rem}

\begin{cor}[A published parameter count certifies the basis]
\label{C:certificate}
In an adapted basis the Hamiltonian-informed ansatz cannot carry more
parameters than the symmetry-filtered one. A reported pair of counts in
which it carries more therefore certifies that the Hamiltonian was not in an
adapted basis, with no calculation required of the reader.
\end{cor}

Table~2 of \cite{he2026} reports both counts for ten molecules, and in five of them the
Hamiltonian-informed ansatz carries more parameters than the symmetry-filtered one: 93
against 75 for NH$_3$, 188 against 65 for CH$_4$, 65 against 49 for N$_2$, 106 against
80 for CO, and 148 against 110 for NaH. By Corollary~\ref{C:certificate} none of those
five pairs is possible in a symmetry-adapted basis, so the Hamiltonian entering that
pipeline was not in one. In the other five rows, HF, LiH, H$_2$O, BeH$_2$ and
C$_2$H$_4$, the two counts agree exactly, which is the equality the containment allows.

Those same five rows, and no others, are the ones whose symmetry-filtered energies the
authors strike through as insufficiently accurate. The agreement is exact across the ten,
and the point group does not sort them. Two of the five have finite non-Abelian groups
and three are linear, but HF, LiH and BeH$_2$ are linear as well, reduced to the same
Abelian subgroups $C_{2v}$ and $D_{2h}$, and their counts satisfy the certificate. A
property of the point group cannot separate molecules that way. A basis adapted in some
runs and not in others can, and the certificate identifies which runs those are before
any energy is compared.

That reading uses the parameter counts alone and assumes nothing about how the runs were
optimised. A second
consequence follows from Corollary~\ref{C:protocol} and is independent of it. Those
authors record their protocol. Every parameter is initialised to zero, all optimisation
uses the BFGS algorithm of SciPy, and sampling and hardware noise are neglected. That is
the protocol of Corollary~\ref{C:protocol}. Its remaining hypotheses are supplied by the
defaults of that implementation, which sets the initial approximate Hessian to the
identity and takes its step lengths from a Wolfe line search, and the line search enforces
the curvature condition $y^{\top}s>0$, which makes the skipped secant update vacuous. The
one hypothesis no default supplies is whether the gradients are analytic, and those
authors do not state it. Under the corollary the filtered and unfiltered
ans\"atze cannot be separated at all, so the two entries of their Table~2,
$1.09\times10^{-4}$ and $2.78\times10^{-2}$, cannot both come from the same Hamiltonian
in adapted orbitals. The authors state explicitly that the theoretical
underpinnings of their method for non-Abelian groups have not been formally
established, and that their conclusions rest on numerical results. That hedge
delimits the claim, and the diagnosis above operates within it.

\begin{prop}[The dynamical Lie algebra of the filtered ansatz is
non-Abelian]
\label{P:dla}
Let $\mathfrak{g}_{\mathrm{Sym}}
=\Lie\bigl(\{\kappa_k:\tau_k\in\mathcal{P}_{\mathrm{Sym}}\}\bigr)$ be the
Lie algebra generated by the retained anti-Hermitian generators under
iterated commutators, which determines the reachable set of the product
ansatz \cite[Ch.~3]{dalessandro2007}. If the filtered pool contains two
single excitations out of a common totally symmetric occupied orbital into
distinct totally symmetric virtual orbitals, then
$\mathfrak{g}_{\mathrm{Sym}}$ is non-Abelian.
\end{prop}

\begin{proof}
One explicit pair of retained generators settles it, and no rank computation
is needed. Let $i$ be a totally symmetric occupied orbital and let $a\neq b$
be totally symmetric virtual orbitals, so that both single excitations
$i\to a$ and $i\to b$ satisfy $\chi_k\equiv1$ and lie in the filtered pool.
Writing $\kappa_1=\hat a^\dagger_a\hat a_i-\hat a^\dagger_i\hat a_a$ and
$\kappa_2=\hat a^\dagger_b\hat a_i-\hat a^\dagger_i\hat a_b$, the terms that
survive in the commutator are the crossed ones, and
\begin{equation}\label{E:commutator}
  [\kappa_1,\kappa_2]=\hat a^\dagger_b\hat a_a-\hat a^\dagger_a\hat a_b
  \neq0,
\end{equation}
the virtual--virtual rotation between $a$ and $b$.
\end{proof}

\begin{rem}[Scope of the certificates]
\label{R:scope_cert}
Both certificates are statements about the orbital basis and carry no verdict on the
merit of the Hamiltonian-informed construction, which recovers the full equivariant
set precisely when the basis is not adapted. The certificate applies directly to
\cite{he2026}. For \cite{he2026a} the reach was checked rather than assumed. That paper
runs no symmetry-filtered comparison of its own, so the certificate has no pair of
counts to test there. It does build its excitation-operator pool on the
Hamiltonian-informed ansatz of \cite{he2026}, so the question of the basis passes to it
by construction, and its optimisation is the same quasi-Newton routine, which places its
runs under Corollary~\ref{C:protocol} on the same footing. For \cite{he2026b} the check
has not been made. Read against \cite{cao2022}, the certificates say that the vanishing one
literature treats as an obstacle is the economy the other exploits, and the orbital
basis decides which of the two readings applies.
\end{rem}

Within a degenerate shell the orbitals are fixed only up to a unitary transformation
acting inside the shell, since any such transformation leaves the Fock eigenvalue
unchanged. Two independent self-consistent-field calculations on the same molecule
therefore return bases that generally differ by an element of that unitary group.
Irrep labels computed from one calculation and applied by orbital index to the other
are labels for a different basis, and the filter built from them is not the filter of
Definition~\ref{D:filters}. The same gauge freedom is recorded independently by
\cite{sakuma2026}, who note that irrep weights measured on hardware depend on the
diagonalisation route when a degenerate shell is partially filled.

What separates this mechanism from a property of the filter is that it can be tested.
A property of the point group cannot depend on the numerical route taken to the
orbitals, while a misalignment between two bases varies with that route by
construction. Section~\ref{S:numerics} measures the spread it produces across
independent orbital determinations of the same molecule.

\section{Numerical corroboration}
\label{S:numerics}

The theorems of Sections~\ref{S:critical} to~\ref{S:certificates} hold
unconditionally once their hypotheses do, so no measurement can strengthen
them. What measurement adds is the size of the quantities they bound and the
geometry at which the hypotheses fail. The section measures the classical amplitude problem in two bases and six
molecules, then the unitary ansatz and its product form, then the boundary in
geometry, and finally the failures reported in the literature. Seven experiments comprising 311 runs stand behind
what follows, each self-contained with its plan, its code, its environment
snapshot and its run log. Two limitations hold throughout: the unitary
measurements use STO-3G alone, and the operator counts are counts of distinct
elementary excitation operators, not two-qubit gates. A pool admits a second
count, over the occurrences of an operator in the product rather than over the
operators themselves, and the two disagree because a same-spin double enters
two singlet parameters with opposite signs. Every count reported here is the
first of the two.

Four conventions hold for every figure of this section and are not repeated in
the captions. Energy axes are logarithmic, so a value at or below the floor
printed on the axis is drawn at that floor, and a difference of exactly zero
cannot be drawn at all; where markers are missing for that reason, the caption
says how many. A dashed horizontal line marks chemical accuracy at
1.6~mHa. Where a tick label carries a second line, that line is the
molecule's full point group. Horizontal offsets and jitter within a group are
cosmetic, and the full-group series is hatched so that it stays separable from
the Abelian one in grayscale.

\subsection{Classical amplitudes under both filters, in two bases}
\label{SS:classical}

Confining coupled-cluster amplitudes to the totally symmetric subspace of the
full point group leaves the correlation energy where it was. Under the full
group the amplitude classes fall from 135to 30for ammonia and from 230to 21for methane in
STO-3G (Table~\ref{Tab:classical}). Across the molecules whose point group is
finite the filter costs at most \num{2.5e-9}~mHa. In 6-31G the same
filter costs at most \num{3.4e-12}~mHa while compressing ammonia from
1325to 277classes and methane from
1890to 144. The standard package
\cite{sun2020} cannot be asked for this. It detects the full group and then
works in the largest Abelian subgroup it supports, reducing ammonia to a
single reflection and methane to three orthogonal twofold axes. We therefore
built the projector from the Cartesian action of each group element on the
molecular orbitals and applied it to the amplitudes at every iteration. Water
is the control in Figure~\ref{F:e0}. Its point group is already Abelian, so
the two filters return identical counts, as they must. The lower panel of
that figure carries the price of the compression the upper panel shows: no
filter moves the correlation energy of any of the three by as much as the
tolerance the coupled-cluster solver converges to.

The amplitudes themselves settle the question before any energy is compared.
Converged and unconstrained, they already satisfy full-group invariance to
\num{1.8e-12}in the singles and \num{8.3e-14}in the doubles, so the filter
removes classes that carry no weight rather than discarding weight the
wavefunction needed. That is the vanishing rule of \v{C}\'arsky et
al.~\cite{carsky1987} measured rather than assumed. A specific published
claim is under test. Table 3 of He et al.~\cite{he2026} is a 6-31G table. Its text
states that in that basis the symmetry-filtered ansatz fails systematically
for non-Abelian molecules, and it names HF, LiH, BeH$_2$ and NH$_3$ as cases
that had succeeded in STO-3G. The filter that ansatz implements is the Abelian
one, and on all four of those molecules at 6-31G it costs at most
\num{2.0e-10}~mHa, ten orders of magnitude below chemical accuracy.
Three of those four molecules are linear, so their full point group is
continuous and only the Abelian filter is under test there, which is
precisely the filter in dispute. All of this is classical coupled cluster and
not a variational quantum eigensolver. The amplitude space of CCSD and the
parameter space of singlet unitary coupled cluster are the same space, which
is why the class counts transfer. Neither the ansatz structure nor the
optimiser transfers with them.

\begin{table}[t]
\centering
\small
\caption{Amplitude classes surviving each filter, and the correlation-energy
cost of imposing it, measured against unconstrained CCSD in the same basis
with a restricted Hartree--Fock reference in the adapted basis. A negative
cost means the constrained calculation returned a marginally lower
correlation energy, which at this magnitude reflects the convergence
threshold rather than an ordering. Groups marked $^{\ast}$ are continuous:
the group-average projector is defined for finite groups only, so for the
three linear molecules the full-group column reports the largest finite
subgroup and coincides with the Abelian column by construction. Those three
rows test the Abelian filter, which is the filter under dispute.}
\label{Tab:classical}
\setlength{\tabcolsep}{3.5pt}
\begin{tabular}{llcrrrrr}
\toprule
& & & \multicolumn{3}{c}{Amplitude classes}
& \multicolumn{2}{c}{Cost (mHa)} \\
\cmidrule(lr){4-6}\cmidrule(lr){7-8}
Basis & Molecule & Group & None & Abelian & Full & Abelian & Full \\
\midrule
STO-3G & HF      & $C_{\infty v}^{\ast}$ & 20   & 11  & 11  & \num{-2.2e-11} & \num{-2.2e-11} \\
       & LiH     & $C_{\infty v}^{\ast}$ & 44   & 20  & 20  & \num{-8.7e-14} & \num{-8.7e-14} \\
       & BeH$_2$ & $D_{\infty h}^{\ast}$ & 90   & 23  & 23  & \num{-4.6e-13} & \num{-4.6e-13} \\
       & H$_2$O  & $C_{2v}$              & 65& 26& 26& \num{3.1e-12}  & \num{3.1e-12}  \\
       & NH$_3$  & $C_{3v}$              & 135& 75& 30& \num{2.6e-12}  & \num{2.5e-9}   \\
       & CH$_4$  & $T_d$                 & 230& 65& 21& \num{1.4e-13}  & \num{1.5e-13}  \\
\midrule
6-31G  & HF      & $C_{\infty v}^{\ast}$ & 495& 178& 178& \num{-4.4e-13} & \num{-4.4e-13} \\
       & LiH     & $C_{\infty v}^{\ast}$ & 189& 85& 85& \num{1.5e-10}  & \num{1.5e-10}  \\
       & BeH$_2$ & $D_{\infty h}^{\ast}$ & 495& 125& 125& \num{2.0e-10}  & \num{2.0e-10}  \\
       & H$_2$O  & $C_{2v}$              & 860& 281& 281& \num{-3.4e-12} & \num{-3.4e-12} \\
       & NH$_3$  & $C_{3v}$              & 1325& 717& 277& \num{9.7e-12}  & \num{3.1e-12}  \\
       & CH$_4$  & $T_d$                 & 1890& 495& 144& \num{5.6e-14}  & \num{-1.8e-13} \\
\bottomrule
\end{tabular}
\end{table}

\begin{figure}[t]
\centering
\includegraphics[width=\linewidth]{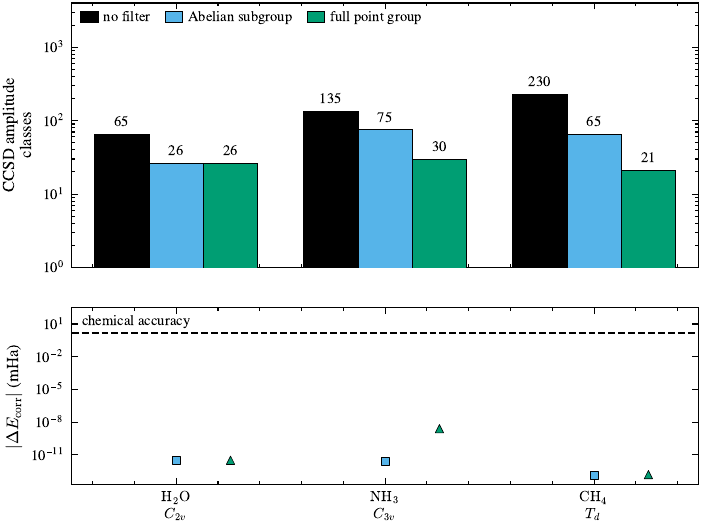}
\caption{Upper panel: CCSD amplitude classes, singles plus doubles, surviving
each filter, for the three molecules of Table~\ref{Tab:classical} whose point
group is finite, in STO-3G from a restricted Hartree--Fock reference in the
adapted basis, each bar labelled with its exact count. The Abelian subgroup is
the one the standard code works in, a proper subgroup of the axis symbol for
NH$_3$ ($C_s$) and CH$_4$ ($D_2$). Lower panel, which carries no legend of its
own: the absolute change each filter causes in the CCSD correlation energy, in
mHa, the Abelian filter as sky blue squares and the full group as green
triangles. Every plotted change lies below
the solver's own convergence tolerance of \num{1e-10}~Ha, so the panel bounds
the cost of each filter rather than resolving it.}
\label{F:e0}
\end{figure}

\subsection{The unitary ansatz}
\label{SS:unitary}

The compression the full point group buys survives the product form a device
would execute, and every measurement in this subsection is STO-3G. We
decomposed each invariant operator on the basis of distinct elementary
excitations and re-exponentiated it as a tied Trotter product, in the natural
and in the reversed order. The decomposition is a rewriting rather than an
approximation, and its largest residual over every molecule and every group
is \num{8.9e-16}. Tying and Trotterising shifts the energy error by at
most 0.0075~mHa (Table~\ref{Tab:unitary}), two orders of magnitude
below chemical accuracy and comparable to the spread between the two
orderings. Neither ordering is preferred, and reporting one of them alone
would understate the uncertainty. The parameters fall from
135to 30for NH$_3$ and from
230to 21for CH$_4$ against unfiltered
unitary coupled cluster, and from 75and 65to
those same counts against the Abelian filter (upper panel of
Figure~\ref{F:e1}). The Abelian comparison is the
honest one, since the Abelian filter is what the literature implements.
Against that same baseline the count of distinct elementary excitation
operators falls from 163to 117for ammonia and from
146to 128for methane. Those two columns are the outer
terms of \eqref{E:hiuccsd} counted. The Abelian entry is the size of
$\mathcal{P}_{\mathrm{Sym}}$ and the full-group entry the size of
$\supp\mathcal{P}_G$, so what the table measures there is how much smaller the
support is than the pool containing it. Water and ethylene carry Abelian
groups, where the two sets coincide, so the full-group column of
Figure~\ref{F:e1} buys them no compression and water returns 48under both filters. In both non-Abelian molecules the parameter count falls by
the larger factor, so the filter removes optimiser dimension first and circuit
size second. Circuit depth was not measured and no correlation-consistent basis
was tested, so nothing here carries a claim about device resources.

One entry of Figure~\ref{F:e1} is not yet accounted for. The Abelian filter can
be built two ways, as a subset of the pool selected by irrep label, which is
what the literature implements, or by applying the symmetric projector; the two
return the same operator count for water and methane and different counts for
ammonia, where the projector construction reads 309distinct
operators against the 163monomials that carry the trivial character
of $C_s$. That pool is therefore not invariant under the subgroup its label
names. Two readings fit, a misaligned reflection plane or a coefficient cut
admitting representation noise, and neither is settled here; the comparisons
drawn above are made against the pool subset throughout.

Proposition~\ref{P:invariance} states an identity that had never been checked
directly. We drew 20parameter vectors per molecule away from the
symmetric subvariety, where the sign flip acts trivially and the test would
be empty, and evaluated both energies through the same circuit. The largest
difference was \num{5.3e-12}~Ha for ammonia and \num{3.6e-15}~Ha for methane, and
the state-level identity held exactly. For methane the word exactly is literal:
of the sixty measurements, fifty-four returned a difference of zero to the last
bit, and all twenty for ammonia were non-zero but at machine precision. The
gradient along the removed
directions stayed below \num{2.2e-10}~Ha, which is Theorem~\ref{T:critical}
measured. Rotating one degenerate orbital pair by thirty degrees while
keeping its irrep labels, the misalignment a second independent
self-consistent-field calculation introduces, raises the same difference to
\num{5.4e-3}and \num{8.1e-3}~Ha (Figure~\ref{F:e4}). That
separation is what distinguishes a measured theorem from an arithmetic
identity, because a violation of the size the literature reports would have
been visible.

\begin{table}[t]
\centering
\small
\caption{Parameters and distinct elementary excitation operators for four
ans\"atze, STO-3G, with the energy error against the exact ground state of
the same Hamiltonian in mHa. The operator count is a count of the distinct
elementary excitation operators a pool touches, not of their occurrences in
the product; it is not a two-qubit gate count, and circuit depth was not
measured. The Abelian row is
built as a subset of the excitation pool selected by irrep label, which is
how the literature implements the filter; the two Trotter columns give the
natural and reversed orderings of the same product, which bound the ordering
dependence.}
\label{Tab:unitary}
\begin{tabular}{llrrrrr}
\toprule
Molecule & Ansatz & Par. & Oper.
& Exact exp. & Trotter, nat. & Trotter, rev. \\
\midrule
H$_2$O / $C_{2v}$ & unfiltered            & 65 & 140 & 0.1203 & --- & --- \\
                  & Abelian               & 26 & 48  & 0.1203 & --- & --- \\
                  & full group            & 26 & 48& 0.1199 & 0.1199 & 0.1144 \\
\midrule
NH$_3$ / $C_{3v}$ & unfiltered            & 135& 315& 0.1392 & --- & --- \\
                  & Abelian               & 75& 163& 0.1392 & --- & --- \\
                  & full group            & 30& 117& 0.1406 & 0.1476 & 0.1415 \\
\midrule
CH$_4$ / $T_d$    & unfiltered            & 230& 560& 0.1942 & --- & --- \\
                  & Abelian               & 65& 146& 0.1942 & --- & --- \\
                  & full group            & 21& 128& 0.1940 & 0.1952 & 0.1944 \\
\bottomrule
\end{tabular}
\end{table}

\begin{figure}[t]
\centering
\includegraphics[width=\linewidth]{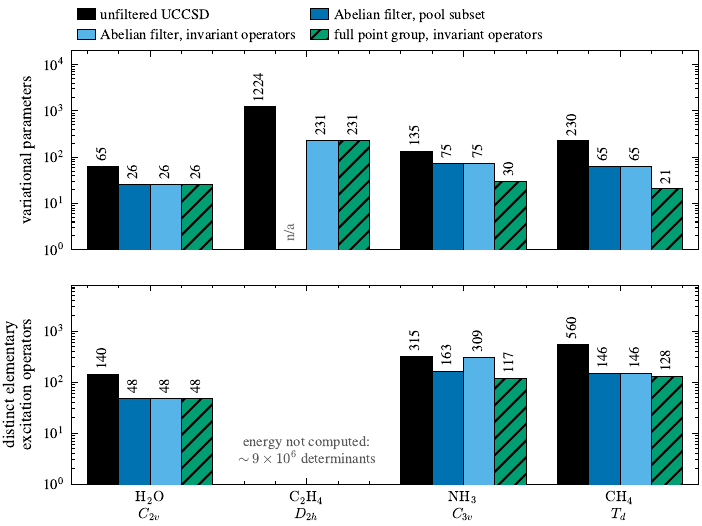}
\caption{Variational parameters (upper panel) and distinct elementary
excitation operators (lower panel) for four ans\"atze across four molecules,
STO-3G, every bar labelled with its exact count. The two Abelian bars are the
same subgroup built the two ways the text distinguishes; they differ only for
ammonia. Water and ethylene are the Abelian controls and return equal counts
under the last two filters. Ethylene was measured by character counting alone,
so it carries no pool-subset bar and no operator counts, its determinant space
being about nine million.}
\label{F:e1}
\end{figure}

\begin{figure}[t]
\centering
\includegraphics[width=0.75\linewidth]{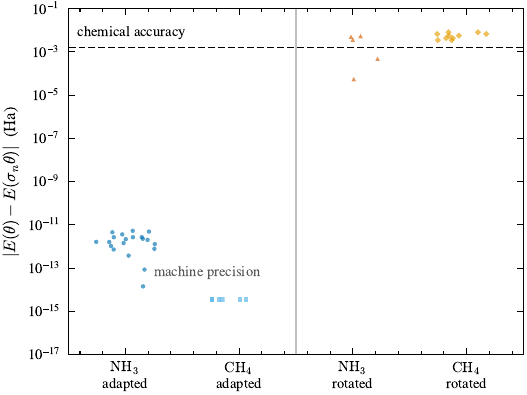}
\caption{Violation of the invariance identity of
Proposition~\ref{P:invariance}, one point per parameter vector and group
element, in STO-3G with the vectors drawn uniformly on $[-0.5, 0.5]$. Left of
the grey rule, the adapted basis: NH$_3$ in $C_s$ (blue circles, 20vectors, one non-trivial element) and CH$_4$ in $D_2$ (sky blue squares,
20vectors, three non-trivial elements). Right of it, the control:
one degenerate orbital pair rotated by thirty degrees with the irrep labels
left stale, five vectors per molecule. Fifty-four of the sixty adapted CH$_4$
measurements are exactly zero and cannot be drawn, as are five of the fifteen
CH$_4$ controls. The closest approach between the two sides is seven decades
for NH$_3$ and twelve for CH$_4$.}
\label{F:e4}
\end{figure}

\subsection{Where the freeness ends}
\label{SS:domain}

Freeness has a domain, and locating its edge is a result of the same standing
as the freeness itself. Ammonia was stretched at fixed H--N--H angle from
0.90to 2.00~\AA{} in STO-3G (Table~\ref{Tab:domain}). Over that scan
the correlation energy grows from 50.7to 439.5~mHa while the gap to
the first excited state collapses from 628.1to 20.3~mHa. The scan
therefore reaches the static-correlation regime it was built to probe. The
Abelian filter costs nothing anywhere in it. Its largest deviation over the
whole scan is \num{7.1e-10}~mHa, while the unfiltered error itself grows a
hundredfold. At the equilibrium geometry, Proposition~\ref{P:enclosure} with
the exhibited point $\theta^\star$ bounds the cost of the full-group filter
by 0.176~mHa for NH$_3$ and 0.194~mHa for CH$_4$ in
STO-3G.

The full-group filter behaves differently along the scan. Its cost stays
negligible near equilibrium and then rises: in the exact-exponential form it
reaches 0.0093mHa at \num{1.40}~\AA, 0.123at
\num{1.60}, 3.85at \num{1.80} and 6.94at
\num{2.00}. In the tied Trotter form the last three distances give
0.734, 6.15and 12.59~mHa
(Figure~\ref{F:e2}). Chemical accuracy is crossed between \num{1.60} and
\num{1.80}~\AA{} in both forms. Where the cost exceeds chemical accuracy, at
\num{1.80} and \num{2.00}~\AA{}, the two realisations of the same
30-parameter space agree within a factor of two, so there the
cost belongs to the invariant subspace and not to the Trotter product. Below
chemical accuracy the two forms carry no such relation, their ratio swinging
by more than two orders of magnitude across the scan while both stay
negligible. Without the exact-exponential control the rise could have been
read as an artefact of the product form. None of this contradicts
Section~\ref{S:critical}. What Theorem~\ref{T:critical} makes critical is the
subvariety cut out by the characters of the Abelian subgroup $N$, at every
geometry and with no domain attached. The boundary located here belongs to the
full-group filter alone. Criticality of the symmetric subvariety is also not
equality of the two infima, and Proposition~\ref{P:enclosure} binds only at
the geometry where $\theta^\star$ was exhibited. At \num{2.20}~\AA{} the
restricted reference stops converging, which is where hypothesis~\ref{H:A2}
ceases to hold. That point is logged as a failure and discarded. Retrying
from a different initial guess would have hidden the boundary the scan was
built to find.

\begin{table}[t]
\centering
\small
\caption{Stretching NH$_3$ at fixed H--N--H angle of $106.7^\circ$ in
STO-3G, 3136 determinants. The last three columns are energy differences from
unfiltered unitary coupled cluster, in mHa, so they are the cost of each
filter at that geometry; the two full-group columns are the same
30-parameter invariant space realised as exact exponentials and
as a tied-parameter Trotter product. The correlation energy and the
$S_0$--$S_1$ gap, which fix the regime each geometry belongs to, are plotted
in the lower panel of Figure~\ref{F:e2}. The \num{2.20}~\AA{} row carries no
measurement, for the reason the text gives.}
\label{Tab:domain}
\begin{tabular}{rrrrr}
\toprule
$r_{\mathrm{NH}}$ (\AA) & UCCSD err.\ (mHa) & Abelian, 75par.
& Full, 30par. & Full, 30par. \\
 & & & exact exp. & Trotter \\
\midrule
0.90   & 0.0992& \num{-1.8e-10} & \num{-4.0e-5} & \num{5.1e-6} \\
1.0124 & 0.1757    & \num{-1.5e-10} & \num{-5.9e-5} & \num{2.1e-3} \\
1.10   & 0.2831    & \num{-1.4e-11} & \num{-2.9e-4} & \num{-2.0e-5} \\
1.25   & 0.6033    & \num{7.1e-12}  & \num{-3.6e-3} & \num{8.8e-3} \\
1.40   & 1.1049    & \num{-7.1e-10} & 0.0093& 0.0859 \\
1.60   & 2.0740    & \num{9.9e-11}  & 0.123& 0.734\\
1.80   & 4.1418    & 0.0            & 3.85& 6.15\\
2.00   & 9.5875& \num{4.3e-10}  & 6.94& 12.59\\
\bottomrule
\end{tabular}
\end{table}

\begin{figure}[t]
\centering
\includegraphics[width=\linewidth]{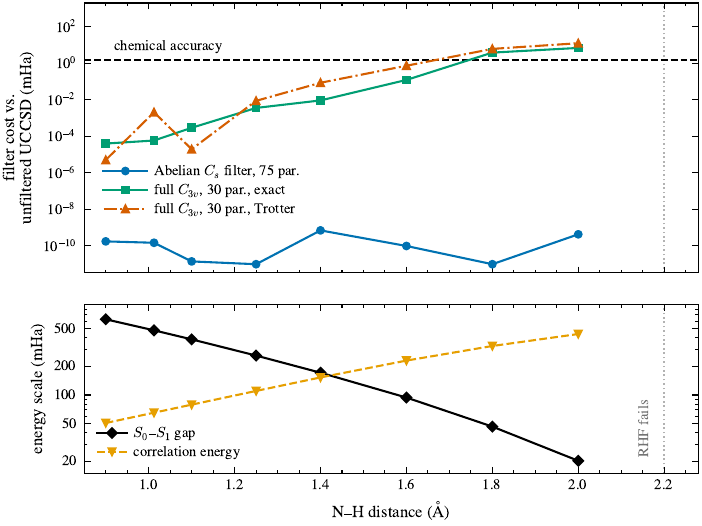}
\caption{Upper panel: absolute energy difference from unfiltered unitary
coupled cluster, against the N--H distance, for NH$_3$ in STO-3G at a fixed
H--N--H angle of $106.7^\circ$ --- the Abelian $C_s$ filter with
75parameters, and the full $C_{3v}$ filter with
30parameters in both realisations, exact exponentials and a
tied Trotter product. The Abelian series stays more than nine decades below
chemical accuracy across the whole scan; both full-group series cross it
between \num{1.60} and \num{1.80}~\AA. Lower panel: the $S_0$--$S_1$ gap and
the absolute correlation energy, which exchange rank between \num{1.40} and
\num{1.60}~\AA{} and so place the far end of the scan in the
static-correlation regime.}
\label{F:e2}
\end{figure}

\subsection{The reported failures}
\label{SS:artifact}

A rule that counts degenerate orbital shells at the mean-field level
separates the molecules reported as failing from those reported as working,
ten of tenin the minimal basis (Figure~\ref{F:e7}). The rule runs no
variational calculation at all. A property of the point group cannot be
predicted by a mean-field shell count. What such a count can predict is a
property of the orbitals rather than of the filter.

Re-running the published two-calculation pipeline for methane in independent
processes returns a filtered energy error anywhere between 15.9and 49.7~mHa (Table~\ref{Tab:artifact}). Over the same processes
the unfiltered error moves by \num{1.7e-4}~mHa, and the same filter
applied with labels taken from the calculation that defines the Hamiltonian
costs 0.194~mHa. The spread appears only when the eigensolver is
given more than one thread (lower panel of Figure~\ref{F:e7}). Single-threaded
runs return 49.722~mHa in ten consecutive processes, and that
value sits above every multi-threaded replicate. It is reported as measured,
since it is what a reader following the published reproducibility settings
would obtain. Ammonia, whose degenerate shells are twofold where those of
methane are threefold, is reproducible at every thread count and returns
29.257~mHa at all of them. The angle between the degenerate subspaces
the two calculations select is fixed there as well, at 60.0degrees. That angle summarises the mismatch, and the deviation is the cost of
the mismatch itself rather than of a rotation by that angle. The
published values this distribution has to account for, 40.8and 27.8~mHa, both fall inside the measured range. A point-group
property cannot depend on how many threads a linear-algebra library was
given, so what these numbers measure is the labelling step of
Section~\ref{S:certificates} and not the filter.

\begin{table}[t]
\centering
\small
\caption{Filtered energy error for CH$_4$ and NH$_3$ in STO-3G, measured
across independent operating-system processes running the published
two-calculation pipeline end to end, grouped by the number of threads given
to the dense linear-algebra library. Replicates are processes rather than
seeds because the quantity that varies is not under the program's control:
it enters through the way the eigensolver resolves a degenerate subspace.
The arm at eight threads for CH$_4$ carries six of the ten replicates
requested and was stopped because two concurrent eight-thread processes
oversubscribe an eight-core machine; the finding rests on the complete arms
at two and four threads.}
\label{Tab:artifact}
\begin{tabular}{lrrrrr}
\toprule
Molecule & Threads & Replicates & Distinct values
& Min (mHa) & Max (mHa) \\
\midrule
CH$_4$ & 1 & 10 & 1  & 49.722& 49.722\\
       & 2 & 10 & 10 & 15.922 & 46.637 \\
       & 4 & 10 & 10 & 18.001 & 44.829 \\
       & 8 & 6  & 6  & 31.841 & 46.614 \\
\midrule
NH$_3$ & 1, 2, 4, 8 & 50 & 1 & 29.257& 29.257\\
\bottomrule
\end{tabular}
\end{table}

\begin{figure}[t]
\centering
\includegraphics[width=\linewidth]{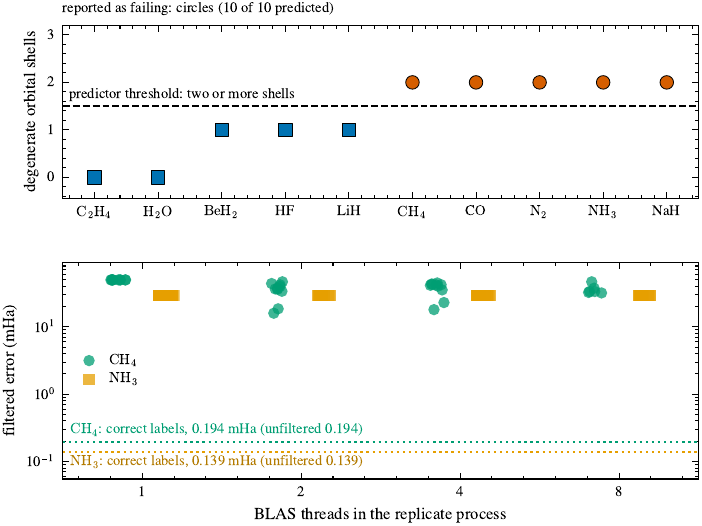}
\caption{Upper panel: degenerate orbital shells at the mean-field level for ten
benchmark molecules in STO-3G, against the outcome reported in the literature,
two orbitals counting as degenerate when their energies agree to
\num{1e-6}~Ha. The threshold rule runs no variational calculation. Lower panel:
the filtered energy error of Table~\ref{Tab:artifact}, in mHa, one point per
independent process, against the threads that process was given; that table
gives the replicate count of each arm. The dotted lines are the
correctly-labelled reference the text quotes, labelled in place. Colour means
the reported outcome above and the molecule below; the two panels share none.}
\label{F:e7}
\end{figure}

A second mechanism of the same class operates before any filter is applied, on
the enumeration of the pool itself. That enumeration is worth writing out,
since it is where the error entered. For NH$_3$ in STO-3G the reference leaves
five occupied and three virtual spatial orbitals. The singlet parameterisation
carries one amplitude per occupied-virtual pair and one per unordered pair of
those, 135in all, and expanding them over spin orbitals
gives the 315distinct elementary excitations of
Table~\ref{Tab:unitary}: thirty singles, sixty same-spin doubles, and
225opposite-spin doubles, the last counting the occupied and the
virtual index pairs as ordered.

Enumerating that last channel as combinations with replacement over occupied
and virtual pairs instead generates 90, so the pool is incomplete
from the start, and running the unrestricted ansatz on that same pool
reproduces the deviation that the filter is then credited with. What a
construction of that kind measures is the pool and not the filter. On a
complete pool Corollary~\ref{C:protocol} makes the two ans\"atze return the
same energy exactly, and the ten digits to which they agree here are that
identity as finite precision realises it. The deviation of 21.8~mHa for
NH$_3$ in STO-3G reported in versions 1 and 2 of this work \cite{dasilva2026}
has that origin.

The reading of the dynamical Lie algebra fails for a separate reason.
Restriction of that algebra to a torus does not imply variational
incompleteness, since reaching the full unitary group of a degenerate shell is
not what the variational problem requires, the target state being totally
symmetric and the filtered generators reaching every state the invariant ones
reach (Theorem~\ref{T:span}, Corollary~\ref{C:reach} and
Remark~\ref{R:orbit}). That algebra is not a torus in
any case. Of the retained generator pairs for that system,
2191of 2775have non-vanishing commutators, which is
the witness Proposition~\ref{P:dla} records in general.

What the section supports is the Abelian filter at no measurable cost wherever
the hypotheses of Section~\ref{S:setup} hold, and the full-group filter at no
measurable cost near equilibrium, together with the geometry beyond which the
latter degrades. It does not support reading the reported non-Abelian failures
as a property of the filter. Two things remain
unmeasured: the unitary ansatz in any basis larger than STO-3G, and any claim
about device-level resources.

\section{Scope}
\label{S:scope}

Every statement proved above rests on hypothesis~\ref{H:A2}, an invariant
reference. It holds for a closed-shell restricted determinant built from
adapted orbitals, which is the setting of the calculations discussed here.
It fails in three places. An open-shell reference is not invariant under the
full group. A target state carrying a different irreducible representation
from the reference, as an excited state does, breaks the parity argument that
Theorem~\ref{T:critical} runs at an arbitrary point of the subvariety. An
orbitally degenerate ground state fails the same way. In those settings the
gradient along a removed direction need not vanish, and the filter can cost
something real. Abandoning the symmetric reference carries a price of its own.
Bertels et al.~\cite{bertels2022} and Tsuchimochi et
al.~\cite{tsuchimochi2022} report independently that symmetry breaking
lengthens the ansatz and degrades the optimisation. The hypothesis is
therefore a methodological choice with evidence behind it, and the boundary it
draws is one the alternative does not remove.

A reader might expect Theorem~\ref{T:critical} to carry with it the positivity
of the curvature transverse to the symmetric subvariety at the optimum, which
would upgrade the critical subvariety to a local minimum. That statement is
absent and may not be inferred from the theorem. We have measured the curvature
at one geometry in one molecule and we do not state it as a result, because we
expect it to be false in general, for the same reason that an unrestricted
Hartree--Fock solution can lie below the restricted one. The equality of the
restricted and the unrestricted infima is absent on the same terms. We do not
prove it, and we expect it to fail in general. What stands in its place is
Proposition~\ref{P:enclosure}, which uses
the variational principle and one exhibited point, and which binds at the
geometry where that point was exhibited. Away from that geometry the enclosure
is silent, and Section~\ref{S:numerics} measures what happens there.

Continuous groups lie outside every theorem stated here. The group-average
projector $|G|^{-1}\sum_g$ is a finite sum, and for $C_{\infty v}$ and
$D_{\infty h}$ the corresponding object is an integral over the group.
Shkolnikov et al.~\cite{shkolnikov2023} meet the same obstruction and reduce
to a finite Abelian subgroup, which is precedent for declaring the reduction
and not for treating the continuous group. This is why the three linear
molecules of Section~\ref{S:numerics} test the Abelian filter only.

The two halves of the measurement have different reach. The classical amplitude
result spans six molecules in two basis sets, one minimal and one split
valence; the unitary result is STO-3G throughout, so nothing here supports a
claim about the unitary ansatz at a larger basis. Two further limits are stated
where the measurements are, at the head of Section~\ref{S:numerics}: what the
operator counts count, and that circuit depth was never measured.

Remark~\ref{R:hisupport} carries the one statement of
Section~\ref{S:certificates} that is measured and not proved, and says what a
proof of the missing direction would have to exclude. Only the inclusion of
Proposition~\ref{P:hisupport} is used.

The classical treatment of point-group symmetry in coupled-cluster theory is
older than the quantum-computing literature that inherited it, and
Remark~\ref{R:priority} records what each of those papers established. We claim
no priority for exploiting non-Abelian point-group symmetry, which
\v{C}\'arsky et al.~\cite{carsky1987} did in 1987, on ammonia and in $C_{3v}$.
The group-average projector is not presented as new either, since it is
Eq.~(6) of \cite{haser1991}. Stanton et al.~\cite{stanton1991} put the saving
at a factor of up to $h^2$ in the group order $h$ and estimated a
factor of about 500 for $T_d$. Nor is the $1/h$ scaling of the retained
amplitude count a discovery here, since Cao et al.~\cite{cao2022} and Greiner
et al.~\cite{greiner2024} both publish it. What is left once those are set
aside is the unitary counterpart, and Section~\ref{S:conclusion} states it.

Four lines of work reach the same object from another side. Vaquero-Sabater et
al.~\cite{vaquero2025} remove from an adaptive ansatz those operators whose
amplitudes come out near zero, which is an empirical test applied once the
optimisation has run. The group filter identifies which operators are exactly
zero before the run, and says why. Greiner et al.~\cite{greiner2024} exploit
the full non-Abelian group on H$_2$O, NH$_3$ and CH$_4$, three of the six
molecules studied here, and what they compress are
increments of a classical many-body expansion, not parameters of a variational
ansatz. They record the $1/h$ bound as one those three systems never formally
attain, which the compression ratios measured here approach from above without
reaching. Adaptive constructions in the manner of Grimsley et
al.~\cite{grimsley2019} are covered by Theorem~\ref{T:span}. Any pool lying in
the adapted span inherits the same invariant span, so whether an adaptive
scheme gains anything over the filter depends on whether its candidate pool
extends past it. Picozzi and Tennyson \cite{picozzi2023,picozzi2026} compress
qubits rather than amplitudes, a different currency, and their counts are
independent of the ones reported here.

\section{Conclusion}
\label{S:conclusion}

Filtering a unitary coupled-cluster ansatz by the Abelian subgroup that
current practice uses costs nothing variationally, and that statement carries
no geometric qualifier. Theorem~\ref{T:critical} makes the subvariety $V$ of
\eqref{E:partition}, which the characters of $N$ define, a critical subvariety
of the energy, so the gradient along every removed direction vanishes at every
point of it and not at the reference alone. Corollary~\ref{C:protocol} turns
that into a statement about the optimisation trajectory.
Filtering by the full point group compresses further than
the Abelian subgroup does, by the larger factor in parameters and by a smaller
one in distinct excitation operators, and there the freeness is measured rather
than proved, holding near equilibrium and degrading beyond it.
Theorem~\ref{T:span} supplies the containment
$\Span\mathcal{P}_G\subseteq\Span\mathcal{P}_{\mathrm{Sym}}$, and
Corollary~\ref{C:reach} carries it to the reachable sets, which makes the
smaller orbit sufficient against the fully invariant pool at every geometry.

Between those two pools sits the set the equivariant one touches.
Proposition~\ref{P:hisupport} places the Hamiltonian-informed pool inside that
support, so in an adapted basis the two constructions the recent literature
treats as alternatives draw from one set, and in every case measured here they
draw the same operators from it. Remark~\ref{R:a1content} separates that set
from the larger one an $A_1$ shell label admits, which is where versions 1 and
2 of this work went wrong.

The consequences reach a working calculation. Coupled-cluster amplitudes can
be filtered by the full point group with the correlation energy retained,
whenever the reference is closed-shell and the orbitals are adapted. A
published pair of parameter counts can be read as a statement about the
orbital basis, by Corollary~\ref{C:certificate}, with no calculation required
of the reader. The freeness of the full-group filter also has a measured
geometric boundary. The stretching scan of Section~\ref{S:numerics} runs from
0.90to 2.00~\AA. The Abelian filter stays free at every geometry on
it, while the full-group filter crosses chemical accuracy of 1.6~mHa
between 1.60 and 1.80~\AA{} and degrades from there.

Stanton et al.~\cite{stanton1991} described the non-Abelian generalisation as
known and unimplemented in 1991, and \v{C}\'arsky et al.~\cite{carsky1987} and
H\"aser et al.~\cite{haser1991} had already carried it out on the classical
side. What this paper adds is the unitary counterpart,
a variational statement valid for degenerate irreducible representations
together with the compression that follows from it.

The unitary result should be repeated in a correlation-consistent basis, where
the virtual space is large enough that the compression ratios measured here
need not carry over. Whether the Hamiltonian-informed pool exhausts the
equivariant support is open, and a proof would have to exclude accidental
vanishing of the integrals instead of arguing from symmetry, which constrains
only the coefficients that must vanish. The open-shell and orbitally degenerate settings, where
hypothesis~\ref{H:A2} fails, are where a real variational cost would first
appear, and where the parity argument has to be replaced by an argument of a
different kind. That is what decides how far the freeness of the filter
extends.

\section*{Declarations}

\paragraph{Data and code availability.}
The measurements reported in Section~\ref{S:numerics} were produced by seven
self-contained experiments, each carrying its plan, its code, its
environment snapshot, its aggregated results and its run log. The
mean-field and coupled-cluster calculations use PySCF~\cite{sun2020} and the
second-quantised operator algebra uses OpenFermion~\cite{mcclean2020}. The seven
experiments, the 311 runs behind them, the aggregated results, the
figures of this paper and a from-scratch re-derivation of the results,
sharing no code with the implementation that produced them, are available at
\url{https://github.com/ldsufrpe/point-group-vqe-free}, under the MIT licence
for the code and CC~BY~4.0 for the data and figures.
\ifversionnote
The repository accompanying versions 1 and 2 of
this work, \texttt{ldsufrpe/\allowbreak nh3-symuccsd-confinement}, reproduces the
superseded result and is being replaced together with this version.
\fi

\paragraph{Competing interests.}
The authors declare no competing interests.

\paragraph{Author contributions.}
L.D.S.\ and M.P.S.\ conceived the study and developed the mathematical
results. L.D.S.\ implemented and ran the numerical experiments. M.Z.\ carried
out the from-scratch re-derivation and verified the numerical claims. All
authors contributed to the writing and approved the final manuscript.

\bibliographystyle{spmpsci}
\bibliography{references}

\end{document}